\documentclass[twocolumn]{article}

\usepackage[english]{babel}
\usepackage{amsmath,amssymb,amsthm}
\usepackage[algo2e,lined,linesnumbered,ruled,noend]{algorithm2e}
\usepackage{graphicx}
\graphicspath{{./pics/}}

\usepackage{float}

\usepackage{subfigure}
\newtheorem{definition}{Definition}
\newtheorem{lemma}{Lemma}
\newtheorem{corollary}{Corollary}
\newtheorem{theorem}{Theorem}

\usepackage{placeins}
\usepackage{url}

\newcounter{observationctr} 
\newenvironment{observation}[1] {
\refstepcounter{observationctr}\vspace{0.2cm}\noindent{\bf Observation \arabic{observationctr}. }#1}{\par}

\newcommand{\depends}[1]{}

\newcommand{\state}[3]{\ensuremath{(#2^{#3})#1}}
\newcommand{\uim}{\state{u}{i}{-}}
\newcommand{\vim}{\state{v}{i}{-}}
\newcommand{\wim}{\state{w}{i}{-}}
\newcommand{\uip}{\state{u}{i}{+}}
\newcommand{\vip}{\state{v}{i}{+}}
\newcommand{\wip}{\state{w}{i}{+}}

\newcommand{\vimm}{\state{v}{(i-1)}{-}}

\newcommand{\uimp}{\state{u}{(i-1)}{+}}

\newcommand{\uipm}{\state{u}{(i+1)}{-}}

\usepackage{tikz}
\usetikzlibrary{shapes}
\usetikzlibrary{arrows,petri,topaths}
\usetikzlibrary{decorations.pathmorphing}
\usetikzlibrary{decorations.markings}
\tikzstyle{vertex}=[circle, draw, inner sep=0pt, minimum size=6pt]
\newcommand{\vertex}{\node[vertex]}

\newcommand{\G}{\mathcal{G}}
\newcommand{\T}{\mathcal{T}}
\newcommand{\F}{\mathcal{F}}

\SetKwFor{ForAll}{forall}{do}{end forall}

 \SetKwComment{com}{//}{}
 \SetKwComment{point}{}{}
 \SetKwFor{UpOnTimer}{New round}{do}{end}
 \SetKwInput{Init}{Initial state}
  \SetKwInput{Sending}{Sending}
    \SetKwInput{Reception}{Reception}
      \SetKwInput{compute}{Computation on a node $v$}
  \SetKwInput{PreCompute}{Pre-compute on a node $v$}
    \SetKwInput{PostCompute}{Post-compute on a node $v$}
    \SetKwInput{Round}{Forever Do}
    \SetKwInput{Initialization}{Initialization}
    \SetKwInput{Round}{Execution of a round (forever loop)}
\SetKwInput{Data}{Each node has}

\let\oldnl\nl
\newcommand{\nonl}{\renewcommand{\nl}{\let\nl\oldnl}}

\newcommand{\comm}[1]{\com{\ \color{gray}#1}}
 
\begin{document}

\title{Maintaining a Distributed Spanning Forest in Highly Dynamic Networks
\thanks{A preliminary version of this work (without proofs of correctness and much revisited since) was presented at the 18th Int. Conf. on Principles of Distributed Systems (OPODIS, 2014).}
}

\author{Matthieu Barjon, Arnaud Casteigts, Serge Chaumette, Colette Johnen, and Yessin M. Neggaz\bigskip\\
LaBRI, CNRS, University of Bordeaux}
\date{}

\maketitle

\begin{abstract}
  Highly dynamic networks are characterized by frequent changes in the availability of communication links. These networks are often partitioned into several components, which split and merge unpredictably. We present a distributed algorithm that maintains a forest of (as few as possible) spanning trees in such a network, with no restriction on the rate of change. Our algorithm is inspired by high-level graph transformations, which we adapt here in a (synchronous) message passing model for dynamic networks. The resulting algorithm
has the following properties: First, every decision is purely local---in each round, a node only considers its role and that of its neighbors in the tree, with no further information propagation (in particular, no wave mechanisms). Second, whatever the rate and scale of the changes, the algorithm guarantees that, by the end of every round, the network is covered by a forest of spanning trees in which 1) no cycle occur, 2) every node belongs to exactly one tree, and 3) every tree contains exactly one root (or token). 
We primarily focus on the correctness of this algorithm, which is established rigorously. While performance is not the main focus, we suggest new complexity metrics for such problems, and report on preliminary experimentation results validating our algorithm in a practical scenario.
\end{abstract}


\section{Introduction}

The current development of mobile and wireless technologies enables direct {\it ad hoc} communication between various kinds of mobile entities, such as vehicles, smartphones, terrestrian robots, flying robots, or satellites. In all these contexts, the set of communication links between entities (network topology) changes continuously. Not only changes are frequent, but in general they are unpredictable and can make the network partitioned at any time. Clearly, the usual assumption of connectivity does not hold here. Also, the classical view of a network whose dynamics corresponds to {\em failures} is no longer suitable in these scenarios, where dynamics is the norm rather than the exception. 

This shift in paradigm impacts algorithms and the definition of problems all together. What does it mean, for instance, to elect a leader in a partitioned network? Is the objective to distinguish a unique global leader, whose leadership materialize over time and space, or is it rather to {\em maintain} a unique leader in each connected component, deleting one when two partitions merge and creating a new one when a partition splits? 
The same remark holds for spanning trees. Should an algorithm construct a unique, global tree whose logical edges survive network intermittence, or should it build and maintain a {\em forest} of trees, each of which spans a (as large as possible) part of the network in a classical way? Both viewpoints make sense, and have been considered {\it e.g.} in~\cite{AE84,CFMS12} (former interpretation) or ~\cite{Awerbuch08,CCGP13} (second interpretation).

In this paper, we focus on the second interpretation, which reflects a variety of scenarios where the expected output of the algorithm should relate to the {\em immediate} configuration of the network ({\it e.g.} several subgroups of robots or drones, each subgroup having a spanning tree for coordination). A particular feature of this type of algorithms is that termination never occurs. More significantly, and perhaps differently to self-stabilization, it may happen that the execution never stabilizes ({\it i.e.}, changes are too frequent to converge to a {\em single} tree per component). This precludes approaches where the computation of a new solution requires the previous computation to have completed, which is an important fact.

The present work is an attempt at understanding what can still be computed (and guaranteed) in terms of spanning trees in such dynamic networks, with no assumptions as to the rate of change, their simultaneity, or global connectivity. In this seemingly chaotic context, we present an algorithm that strives to maintain as few trees per components as possible, while {\em always} guaranteeing key properties.

\subsection{Related work}

Several works have addressed the spanning tree problem in dynamic networks, with different goals and assumptions. Burman and Kutten~\cite{asynchronous} and Kravchik and Kutten~\cite{synchronous} consider a self-stabilizing approach where the legal state corresponds to having a (single) minimum spanning tree and the faults are topological changes. The strategy consists in recomputing the entire tree when a change occurs. This general approach, sometimes called the ``blast away'' approach, is meaningful if stable periods of time exist, which is not the case in (unrestricted) highly dynamic networks.

A number of spanning tree algorithms use random walks for their elegance and simplicity, as well as for their inherent locality. In particular, approaches that involve multiple coalescing random walks allow for uniform initialization (each node starts in the same state) and topology independence (same strategy whatever the graph). Pionneering studies involving such processes include Bar-Ilan and Zernik~\cite{BZ89} (for the problem of election and spanning tree), Israeli and Jalfon~\cite{IJ90} (mutual exclusion), and Chapter 14 of Aldous and Fill~\cite{AF02} (general analysis). 

The principle of using coalescing random walks to build spanning trees in mildly dynamic networks was used by Baala et al.~\cite{mosbah-tree} and Abbas et al.~\cite{Baala03}, where tokens are annexing territories gradually by capturing each other. Regarding dynamicity, both algorithms require the nodes to know (an upper bound on) the cover time of the random walk, in order to regenerate a token if they have not been visited for some time. Besides the strength of this assumption (akin to knowing the number of nodes $n$, or the size of components in our case), the efficiency of the timeout approach decreases dramatically with the rate of topological changes. In particular, if they are more frequent than the cover time (itself in $O(n^3)$), then the tree is constantly fragmented into ``dead'' root-less ({\it i.e.} leader-less) pieces.

Another algorithm based on random walks is proposed by Bernard et al.~\cite{BBS13}. Here, the tree is constantly redefined as the token moves (in a way that reminds the snake game). Since the token moves only over present edges, those edges that have disappeared are naturally cleaned out of the tree as the walk proceeds. Hence, the algorithm can tolerate failure of the tree edges. However it still suffers from detecting the disappearance of tokens using timeouts based on the cover time, which as we have seen, suits only slow dynamics.

A recent work by Awerbuch et al.~\cite{Awerbuch08} addresses the maintenance of {\em minimum} spanning trees in dynamic networks. The paper shows that a solution to the problem can be updated after a topological change using $O(n)$ messages (and same time), while the $O(m)$ messages of the ``blast away'' approach was thought to be optimal. (This demonstrates, incidentally, the revelance of {\em updating} a solution rather than recomputing it from scratch in the case of minimum spanning trees.) The algorithm has good properties for highly dynamic networks. For instance, it considers as natural the fact that components may split or merge perpetually. Furthermore, it tolerates new topological events while an ongoing update operation is executing. In this case, update operations are enqueued and consistently executed one after the other. While this mechanism allows for an arbitrary number of topological events {\em at times}, it still requires that such burst of changes are only episodical and that the network remains eventually stable for (at least) a linear amount of time in the number of nodes, in order for the update operations to complete and thus the logical tree to be consistent with physical reality.

All the aforementioned algorithms either assume that {\em global update} operations (e.g. wave mechanisms) can be performed regularly, or that some node can collect {\em global information} about the tree structure. As far as dynamics is concerned, this forbids arbitrary and ever going changes to occur in the network. 
 

\subsection{A high-level (graph-level) mechanism.}
\label{sec:principle}


A high-level graph scheme was proposed in~\cite{CCGP13} for the maintenance of a spanning forest (not necessarily minimum) in unrestricted dynamic networks, using a coarse grain interaction model inspired from graph relabeling systems~\cite{GRS01} (bearing some common traits with so-called population protocols~\cite{AAD+06}). It can be described informally as follows. Initially every node hosts a token and is the {\em root} of its own individual tree. Whenever two roots/tokens are located at both {\em endpoints} of a same edge (see merging rule on Figure~\ref{fig:scheme}), one of them is destroyed and the underlying node selects the other as parent: both trees (of arbitrary size) are merged locally and instantly. In absence of merging opportunity, the tokens execute a random walk within their own tree in the hope for (farther) merging opportunities (see circulation rule on Figure~\ref{fig:scheme}). As they circulate, the tokens flip (again, locally) the parent-child relations so that a directed path from any node in the tree towards its root is maintained. The fact that the random walk takes place {\em within the tree} (as opposed to the whole network) is crucial for this property. In fact, this simple feature is what enables to recover a consistent state immediately after an edge of the tree has disappeared. Indeed, it suffices for the child side of the lost edge to regenerate a new token/root, while being safe that no other node in the tree can do so (see reparation rule on Figure~\ref{fig:scheme}). In conclusion, this scheme allows for {\em all} operations to be handled in a purely localized fashion (let apart global convergence).

\begin{figure}[h]
  \centering
  \subfigure[Merging rule]{
    \begin{tikzpicture}[scale=.8]
      \clip (-2,-1.2) rectangle (1,.2);
      \tikzstyle{every node}=[draw,circle,fill=black!80,inner sep=1.8pt]
      \path (-1,0) node (v11) {};
      \path (.1,0) node (v12) {};
      \path (-1,-1) node (v21) {};
      \path (.1,-1) node[fill=white] (v22) {};
      \draw[dashed] (v11)-- coordinate[midway](mid1) (v12);
      \draw[very thick,<-] (v21)-- coordinate[midway](mid2)(v22);
      \draw[thick, gray!70, shorten >=6pt, shorten <=6pt] (mid1) edge[->] (mid2);
      \tikzstyle{every node}=[font=\footnotesize,above=2pt]
    \end{tikzpicture}
  }\hspace{1cm}
  \subfigure[Circulation rule]{
    \begin{tikzpicture}[scale=.8]
      \clip (-2.3,-1.2) rectangle (1.3,.2);
      \tikzstyle{every node}=[draw,circle,fill=black!80,inner sep=1.8pt]
      \path (-1,0) node (v11) {};
      \path (.1,0) node[fill=white] (v12) {};
      \path (-1,-1) node[fill=white] (v21) {};
      \path (.1,-1) node (v22) {};
      \draw[very thick,<-] (v11)--(v12);
      \draw[very thick,->] (v21)--(v22);
      \draw[thick, gray!70, shorten >=6pt, shorten <=6pt] (mid1) edge[->] (mid2);
      \tikzstyle{every node}=[font=\footnotesize,above=2pt]
    \end{tikzpicture}
  }\hspace{1cm}
  \subfigure[Regeneration rule]{
    \begin{tikzpicture}[scale=.8]
      \clip (-3.3,-1.2) rectangle (.3,.2);
      \tikzstyle{every node}=[draw,circle,fill=black!80,inner sep=1.8pt]
      \path (-2,0) node[fill=white] (v11) {};
      \path (-2,-1) node (v21) {};
      \tikzstyle{every node}=[]
      \path (-.9,0) node[gray] (v12) {\hspace{-1.1cm} \LARGE $\times$};
      \draw[very thick, ->] (v11)--(v12);
      \draw[thick, gray!70, ->, shorten >=4pt, shorten <=4pt] (v11)--(v21);
    \end{tikzpicture}
  }
  \caption{\label{fig:scheme} Spanning forest principle (high-level representation). {\it Black nodes are those having a token. Black directed edges denote child-to-parent relationships. Gray vertical arrows represent transitions.}}
\end{figure}
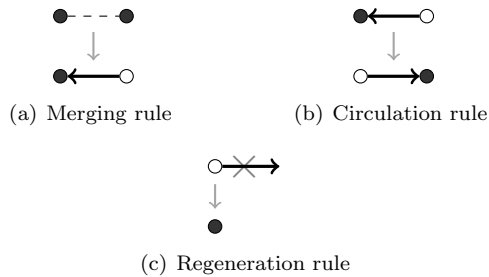

More precisely, at a graph level, this scheme guarantees that the network remains covered by a spanning forest at any time, in which 1) no cycle ever appears, 2) maximal subtrees always are directed rooted trees (with a token at the root), and 3) every node always belongs to such a tree, and so, whatever the rate and scale of topological changes. As to performance, analyzing it requires first to define what metric is relevant in this context. It is not expected that the rate of changes allows {\em any} algorithm to converge towards a single tree per connected component (which is, in a sense, the optimal state in such problem). Before such concerns, a more important question remained to be answered as to whether such a mechanism could be implemented in more conventional message passing models.

\subsection{Our contribution.}

We present the first adaptation of the above mechanism into the synchronous message passing model from~\cite{KLO10}. Due to the loss of atomicity (in particular, the loss of exclusivity) in the interaction, the algorithm turns out to be much more sophisticated than its graph-level counterpart. While still reflecting the same abstract principle, it faces problems that require conceptual differences. In particular, the original model prevented (conveniently) a node to select a parent at the same time as it is itself selected as parent by another node, thereby making cycle avoidance straightforward. One of the ingredients in the new algorithm to circumvent this type of problem is an original technique (referred to as the {\em unique score} technique) that consists of maintaining, network-wide, a set of {\tt score} variables that always remain a permutation of the set of nodes IDs. This mechanism allows us to break symmetry and avoid the formation of cycles in a context where IDs alone did not suffice. (We believe this technique is of independent interest.)
The paper is organized as follows. In Section~\ref{sec:model}, we present the synchronous message passing model from~\cite{KLO10}, slightly adapted (in an equivalent way) and notations that we use throughout the paper. Then, Section~\ref{sec:algorithm} presents the algorithm, whose correctness analysis is developped through Section~\ref{sec:correctness}. Finally, Section~\ref{sec:simulation} discussed some aspects regarding to performance, which includes preliminary experimental results, which can be seen as (partial) practical validation of our algorithm.

\section{Model and notations}
\label{sec:model}

The network is represented by a graph sequence $\G=(G_1,G_2,\dots)$, such that $G_i=(V,E_i)$, where $V$ is a fixed set of vertices and $E_i$ is a dynamically changing set of undirected edges. 
Following Kuhn et al.~\cite{KLO10}, we consider a synchronous (rounded) computational model, where in each round $i$, the set of edges $E_i$ determines what nodes communicate. At the beginning of each round, each node sends a message that was prepared at the end of the previous round. This message is sent to all its neighbors in $E_i$, although the list of these neighbors is a priori {\em unknown} to the node. Then, it receives all messages sent by its neighbors (in the same round), and finally computes its new state and its message for the next round. Due to the reciprocity of undirected links, a node {\em can} determine upon reception which nodes have received its own message. In summary, each round corresponds to three phases {\tt (send, receive, compute)}, which corresponds to a rotation of the original model of~\cite{KLO10} where the phases are {\tt (compute, send, receive)}. This adaptation is not necessary, but makes the expression of our algorithm and of its correctness simpler. In particular, correctness predicates are satisfied by the {\em end} of each round (as opposed to the middle of each round if the original model had been used).


The nodes possess unique identifiers taken from a totally ordered set; that is, for any two nodes $u$ and $v$, it either holds that $ID(u)>ID(v)$ or $ID(v)>ID(u)$. A node can specify what neighbor its message is intended to (although all neighbors will receive it) by setting the {\tt target} field of that message. Symmetrically, the $ID$ of the emitter of a message can be read in the {\tt sender} field of that message. Since the edges are undirected, if $u$ receives a message from $v$ at round $i$, then $v$ also receives a message from $u$ at round $i$. We call this property the principle of \textit{reciprocity}.

Globally, the progress of the execution is represented as a sequence of {\em configurations} $(C_0, C_1, C_2, ... , C_i)$, where each $C_i$ captures the state of all nodes at the end of round $i$  (except for $C_0$, the initial state). We use interchangeably the terminology ``after round $i$'' and ``at/by the end of round $i$'', and similarly for ``before round $i$'' and ``at the beginning of round $i$''.

\section{The Algorithm}
\label{sec:algorithm}

In this section, we present a message passing algorithm which adapts (``implements'' in the theoretical sense) the spanning forest mechanism described in Section~\ref{sec:principle} into the synchronous model from~\cite{KLO10}. We first describe the variables present at each node, then the structure of a message, and finally the algorithm itself with both an informal description and detailed listings of pseudo-code.

\subsection{State variables}
\label{sec:variables}

Besides the {\tt ID} variable, which we assume is externally initialized, each node has a set of variables which reflect its situation in the tree:
\texttt{status} accounts for the possession of a token ({\tt T} if it has a token, {\tt N} if it does not);
\texttt{parent} contains the {\tt ID} of this node's parent ($\bot$ if it has none); \texttt{children} contains the set of this node's children ($\emptyset$ if it has none).
Observe that both variables {\tt status} and {\tt parent} are somewhat redundant, since in the spanning forest principle (see Section~\ref{sec:principle}) the possession of a token is equivalent to being a root. Our algorithm enforces this equivalence, yet, keeping both variables separated simplifies the description of the algorithm and our ability to think of it intuitively.
Variable {\tt neighbors} contains the set of nodes from which a message was received in the last reception. These neighbors may or may not belong to the same tree as the current node. Variable {\tt contender} contains the {\tt ID} of a neighbor that the current node considers selecting as parent in the next round (or $\bot$ if there is no such node). Finally, the variable {\tt score} is the main ingredient of our cycle-avoidance mechanism, whose role is described below. 

\subsubsection{Initial values:} All the nodes are uniformly initialized. They are initially the root of their own individual tree ({\it i.e.} $\texttt{status}=T$, $\texttt{parent}=\bot$, and ${\tt children}=\emptyset$). They know none of their neighbors (${\tt neighbors}=\emptyset$), have no contenders (${\tt contender}=\bot$), and their {\tt score} is set to their own {\tt ID}.

\subsection{Structure of a message (and associated variables)}

Messages are composed of a number of fields: {\tt sender} is the ID of
the sending node; {\tt senderStatus} its status (either {\tt T} or
{\tt N}); and {\tt score} its score when the message was prepared. The
field {\tt action} is one of $\{FLIP,SELECT,HELLO\}$. Informally,
$SELECT$ messages are sent by a root node to another root node to
signify that it ``adopts'' it as a parent (merging operation); $FLIP$
messages are sent by a root node to circulate the token to one of its
children (circulation operation); $HELLO$ messages are sent by a node
by default, when none of the other messages are sent, to make its
presence and status known by its neighbors. Finally, {\tt target} is
the ID of the neighbor to which a FLIP or a SELECT message are
intended ($\bot$ for HELLO messages).

Received messages are stored in a variable \texttt{mailbox}, which is
a map collection whose {\em keys} are the senders ID ({\it i.e.,} a
message whose sender ID is $u$ can be accessed as {\tt mailbox[$u$]}).
In each round, the algorithm makes use of a {\tt RECEIVE()} function
that clears the mailbox and fill it with all the messages received in
that round (one for each physical neighbor). A node can thus update
the set of its neighbors by fetching the {\em keys} of its mailbox.
Similarly, it can eliminate from its list of children those nodes which
are no more neighbor.

As mentioned above, every node prepares at the end of a round the message to
be sent at the beginning of the next round. This message is stored in
a variable \texttt{outMessage}. We allow the short hand $\texttt{m} \gets (a, b, c, d, e)$ to define a new message $m$ whose emitter is node $a$ (with status $b$ and score $e$); target is node $d$; and action is $c$.

\subsubsection{Initial values:} The mailbox is initially empty ({\tt mailbox} $=\emptyset$) and \texttt{outMessage} is initialized to the tuple $(ID, T, HELLO, \bot, ID)$.

\subsection{Description of the algorithm}

The algorithm implements the general scheme presented in
Section~\ref{sec:principle}. In this Section we explain how each of
the three core operations ({\em merging}, {\em circulation}, {\em
  regeneration}) is implemented. Then we discuss the specificities of
the merging operation in more detail and the problems that arise due
to its entanglement with the circulation operation, a fact due to the
loss of atomicity in the message passing model. The resulting solution
is substantially more sophisticated than its original scheme, and yet
it faithfully reflects the same high-level principle. Let us start
with some generalities. In each round, each node broadcasts to its
neighbors a message containing, among others, its status ({\tt T} or
{\tt N}) and an action (SELECT, FLIP, or HELLO). Whether or not the
message is intended to a specific $target$ (which is the case for
SELECT and FLIP messages), all the nodes who receive it can possibly
use this information for their own decisions. More generally, based on
the received information and the local state, each node computes at
the end of the round its new status and the local structure of its
tree (variables \texttt{children} and \texttt{parent}), then it
prepares the next message to be sent. We now describe the three
operations. Throughout the explanations, the reader is invited to
refer to Figure~\ref{fig:example}, where an example of execution
involving all of them is shown. All details are also given in the
listings of Algorithm~\ref{algo:main} and~\ref{algo:functions}.

\begin{figure*}[h]
  \centering
  \subfigure[Round $1$]{
    \begin{tikzpicture}[yscale=1.1,xscale=1]
      \vertex[fill] (4) at (1,1) [label=above:$4$] {};
      \vertex[fill] (1) at (0,1) [label=above:$1$] {};
      \vertex[fill] (2) at (1,0) [label=below:$2$] {};
      \vertex[fill] (3) at (0.3,0.3) [label=below:$3$] {};
      \vertex[fill] (6) at (2,0.5) [label=above:$8$] {};
      \vertex[fill] (5) at (2,-0.7) [label=below:$5$] {};
      \vertex[fill] (7) at (2.5,-1.5) [label=below:$6$] {};
      \vertex[fill] (8) at (1.5,-1.2) [label=below:$7$] {};
      \path[gray]
      (1) edge node[auto] {$s \rightarrow$} (4)  
      (6) edge node[sloped,above, midway] {$s \rightarrow$} (4)
      (2) edge node[sloped,below, midway] {$s \rightarrow$} (6)
      (2) edge (3)
      (4) edge node[sloped,below, midway] {$s \rightarrow$} (3)
      (1) edge (3)
      (4) edge (2)
      (5) edge (7)
      (8) edge node[sloped,below, midway] {$\leftarrow s$} (7)
      (8) edge node[sloped,above, midway] {$\leftarrow s$} (5)
      ;
    \end{tikzpicture}
  }\hspace{40pt}
  \subfigure[Round $2$]{
    \begin{tikzpicture}[yscale=1.1,xscale=1]
      \vertex[] (4) at (1,1) [label=above:$4$] {};
      \vertex[] (1) at (0,1) [label=above:$1$] {};
      \vertex[] (2) at (1,0) [label=below:$2$] {};
      \vertex[] (3) at (0.3,0.3) [label=below:$3$] {};
      \vertex[fill] (6) at (2,0.5) [label=above:$8$] {};
      \vertex[] (5) at (2,-0.3) [label=below right:$5$] {};
      \vertex[] (7) at (2.5,-1.5) [label=below:$6$] {};
      \vertex[fill] (8) at (1.5,-1.2) [label=below:$7$] {};
      \path[gray]
      (1) edge (4)
      (6) edge (4)
      (2) edge node[sloped,above, midway] {$\leftarrow f$ } (6)
      (2) edge (3)
      (4) edge (3)
      (1) edge (3)
      (4) edge (2)
      (5) edge (7)
      (8) edge (7)
      (8) edge  node[sloped,above, midway] {$f \rightarrow$ }  (5)
      (5) edge node[sloped,above, midway] {}   (6)
      (5) edge (2)
      ;

      \draw[very thick,->] (4) -- (6);
      \draw[very thick,->](3) -- (4);
      \draw[very thick,->](1) -- (4);
      \draw[very thick,->](2) -- (6);
      \draw[very thick,->](5) -- (8);
      \draw[very thick,->](7) -- (8);
    \end{tikzpicture}
  }\hspace{40pt}
  \subfigure[Round $3$]{
    \begin{tikzpicture}[yscale=1.1,xscale=1]
      \vertex[] (4) at (1,1) [label=above:$4$] {};
      \vertex[] (1) at (0,1) [label=above:$1$] {};
      \vertex[fill] (2) at (1,0) [label=below:$8$] {};
      \vertex[] (3) at (0.3,0.3) [label=below:$3$] {};
      \vertex[] (6) at (2,0.5) [label=above:$2$] {};
      \vertex[fill] (5) at (2,-0.3) [label=below right:$7$] {};
      \vertex[] (7) at (2.5,-1.5) [label=below:$6$] {};
      \vertex[] (8) at (1.5,-1.2) [label=below:$5$] {};
      \path[gray]
      (1) edge (4)
      (6) edge (4)
      (2) edge node[sloped,above, midway] {$f \rightarrow$ } (6)
      (2) edge (3)
      (4) edge (3)
      (1) edge (3)
      (4) edge (2)
      (5) edge (7)
      (8) edge (7)
      (8) edge (5)
      (6) edge node[sloped,below, midway] {$s \rightarrow$ } (5)
      (2) edge (5)
      ;
      
      \draw[very thick,->] (4) -- (6);
      \draw[very thick,->](3) -- (4);
      \draw[very thick,->](1) -- (4);
      \draw[very thick,<-](2) -- (6);
      \draw[very thick,<-](5) -- (8);
      \draw[very thick,->](7) -- (8);
    \end{tikzpicture}
  }

  \subfigure[Round $4$]{
    \begin{tikzpicture}[yscale=1.1,xscale=1]
      \vertex[] (4) at (0.6,1.3) [label=above:$4$] {};
      \vertex[] (1) at (-0.5,1.1) [label=above:$1$] {};
      \vertex[] (2) at (1.2,0) [label=below:$2$] {};
      \vertex[] (3) at (0.3,0.3) [label=below:$3$] {};
      \vertex[fill] (6) at (2,0.5) [label=above:$8$] {};
      \vertex[] (5) at (2,-0.3) [label=below right:$7$] {};
      \vertex[] (7) at (2.5,-1.5) [label=below:$6$] {};
      \vertex[] (8) at (1.5,-1.2) [label=below:$5$] {};
      \path[gray]
      (1) edge (4)
      (2) edge (6)
      (2) edge (3)
      (4) edge (3)
      (1) edge (3)
      (5) edge (7)
      (8) edge (7)
      (8) edge (5)
      (6) edge node[sloped,below, xshift=-2pt,midway] {$\leftarrow f$ }  (5)
      (2) edge (5)
      ;
      \draw[dashed,-] (4) --  (2);
      \draw[very thick, dashed,->] (4) -- node[sloped,above, midway] {$\times \leftarrow$ }   (6);
      \draw[very thick,->] (5) -- (6);
      \draw[very thick,->](3) -- (4);
      \draw[very thick,->](1) -- (4);
      \draw[very thick,->](2) -- (6);
      \draw[very thick,<-](5) -- (8);
      \draw[very thick,->](7) -- (8);
    \end{tikzpicture}
  }\hspace{25pt}
  \subfigure[Round $5$]{
    \begin{tikzpicture}[yscale=1.1,xscale=1]
      \vertex[fill] (4) at (0.6,1.3) [label=above:$4$] {};
      \vertex[] (1) at (-0.5,1) [label=above:$1$] {};
      \vertex[] (2) at (1,0) [label=below:$2$] {};
      \vertex[] (3) at (0.3,0.3) [label=below:$3$] {};
      \vertex[] (6) at (2,0.5) [label=above:$7$] {};
      \vertex[fill] (5) at (2,-0.3) [label=below right:$8$] {};
      \vertex[] (7) at (2.5,-1.5) [label=below:$6$] {};
      \vertex[] (8) at (1.5,-1.2) [label=below:$5$] {};
      \path[gray]
      (1) edge node[sloped,above, midway] {$\leftarrow f$ }   (4)
      (2) edge (6)
      (2) edge (3)
      (4) edge (3)
      (1) edge (3)
      (5) edge (7)
      (8) edge (7)
      (8) edge node[sloped,above, midway] {$\leftarrow f$ }   (5)
      (6) edge (5)
      (2) edge (5)
      ;
      
      \draw[very thick,<-] (5) -- (6);
      \draw[very thick,->](3) -- (4);
      \draw[very thick,->](1) -- (4);
      \draw[very thick,->](2) -- (6);
      \draw[very thick,<-](5) -- (8);
      \draw[very thick,->](7) -- (8);
    \end{tikzpicture}
  }\hspace{25pt}
  \subfigure[Round $6$]{
    \begin{tikzpicture}[yscale=1.1,xscale=1]
      \vertex[] (4) at (0.6,1.3) [label=above:$1$] {};
      \vertex[fill] (1) at (-0.5,1) [label=above:$4$] {};
      \vertex[] (2) at (1,0) [label=below:$2$] {};
      \vertex[] (3) at (0.3,0.3) [label=below:$3$] {};
      \vertex[] (6) at (2,0.5) [label=above:$7$] {};
      \vertex[] (5) at (2,-0.3) [label=below right:$8$] {};
      \vertex[] (7) at (2.5,-1.5) [label=below:$6$] {};
      \vertex[fill] (8) at (1.5,-1.2) [label=below:$5$] {};
      \path[gray]
      (1) edge node[sloped,above, midway] {$f \rightarrow$ } (4)
      (2) edge (6)
      (2) edge (3)
      (4) edge (3)
      (1) edge (3)
      (5) edge (7)
      (8) edge (7)
      (8) edge node[sloped,above, midway] {$f \rightarrow$ } (5)
      (6) edge (5)
      (2) edge (5)
      ;
      
      \draw[very thick,<-] (5) -- (6);
      \draw[very thick,->](3) -- (4);
      \draw[very thick,<-](1) -- (4);
      \draw[very thick,->](2) -- (6);
      \draw[very thick,->](5) -- (8);
      \draw[very thick,->](7) -- (8);
    \end{tikzpicture}
  }

  \caption{Example of execution of the algorithm which illustrates all types of operations: parent selection ($s\rightarrow$), token circulation ($f\rightarrow$), and tree disconnection ($\times\leftarrow$). {\it The first two symbols represent FLIP or SELECT messages to be sent in the \underline{next} round. Black (resp. white) nodes are those (not) having a token at the \underline{beginning} of the round. Tree edges are represented by bold directed edges. Dash edges have just disappeared.}}
  \label{fig:example}
\end{figure*}
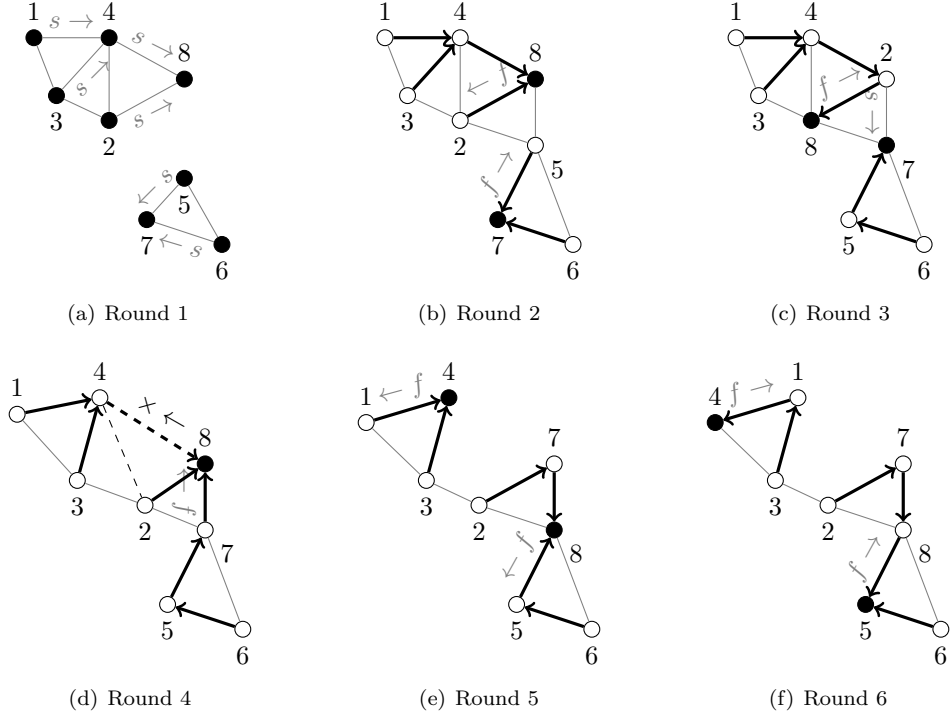

\subsubsection{Merging:} 
If a root ({\it i.e.} a node having a token), say $v$, detects the
existence of a neighbor root with higher {\tt score} than its own,
then it considers that node as a possible \texttt{contender}, {\it
  i.e.} as a node that it might select as a parent in the next round.
If several such roots exist, then the one with highest score, say $u$,
is chosen. At the beginning of the next round, $v$ sends a $SELECT$
message to $u$ to inform it that it is its new parent. Two cases are
possible: either the considered edge is still present in that round,
or it disappeared in-between both rounds. If it is still present, then
$u$ receives the message and adds $v$ to its children list, among others
(Line~\ref{line:addChild}). As for $v$, it sets its {\tt parent}
variable to $u$ and its {\tt status} to N
(Lines~\ref{line:testadoptparent} and~\ref{line:becomeOrdinary}). If
the edge disappeared, then $u$ does not receive the message, which is
lost. However, due to the reciprocity of message exchange, $v$ does
not receive a message from $u$ either and thus simply does not
executes the corresponding changes. By the end of the round, either
the trees are properly merged, or they are properly separated.

\subsubsection{Circulation:} If a root $v$ does not detect another
root with higher score, then it selects one of its children at random,
if it has any (see Line~\ref{line:prepareFlip}), otherwise it simply
remains root. Randomness is not a strict requirement of our algorithm
and replacing it with any deterministic strategy would not affect
correctness of the algorithm. Once the child is chosen, say $u$, the
root prepares a FLIP message intended to $u$, and sends it at the
beginning of the next round. Two cases are again possible, whether or
not the edge $\{u,v\}$ is still present in that round. If it is still
present, then $u$ receives the message, it updates its status and adds
$v$ to its children list, among others (Lines~\ref{line:becomeRoot}
and Line~\ref{line:addChild}). As for $v$, it sets its {\tt parent}
variable to $u$ and its {\tt status} to N
(Lines~\ref{line:testadoptparent} and~\ref{line:becomeOrdinary}). If
the edge disappeared, then $v$ can detect it as before simply does not
executes the corresponding changes. Node $u$, on the other hand,
detects that the edge leading to its current parent disappeared, thus
it regenerates a token (discussed next). Notice that in the absence of
a merging opportunity, a node receiving the token in round $i$ will
immediately prepare a FLIP message to circulate the token in the next
round. Unless the tree is composed of a single node, the tokens are
thus moved in each round. In order for them to remain detectable in
this case, the status announced in $FLIP$ messages is {\tt T} (whereas
it is {\tt N} for $SELECT$ messages).

\SetKwRepeat{Repeat}{repeat}{}%
\begin{algorithm2e*}[h]
  
    \BlankLine
    \medskip

  \Repeat{}{
  \texttt{SEND(outMessage);}
  \BlankLine

  \texttt{mailbox $\gets$ RECEIVE();}\hfill\comm{Received messages, indexed by sender ID}
  \BlankLine

  \texttt{neighbors $\gets$ mailbox.keys();}\hfill\comm{All the senders IDs} \label{line:neighbors}
  \texttt{children} $\gets$ \texttt{children} $\cap$ \texttt{neighbors}\medskip\\   \label{line:intersection}

    \BlankLine
    
    \comm{Regenerates a token if parent link is lost}\smallskip
    \If{\texttt{status}=N $\wedge$ \texttt{parent} $\not \in$ \texttt{neighbors}}{
      \texttt{BECOME\_ROOT()};\label{line:become}\\
    }
    \BlankLine
    \comm{Checks if the outgoing FLIP or SELECT (if any) was successful}\smallskip
    
    \If{\texttt{outMessage.action $\in$ \{FLIP,SELECT\} $\wedge$ outMessage.target $\in$ neighbors\label{line:testadoptparent}}}{
      \texttt{ADOPT\_PARENT(outMessage)}\label{line:becomeOrdinary}
    }
    \BlankLine

    \comm{Processes the received messages}\smallskip    
    
    \texttt{contender} $\gets \bot$;\\
    \texttt{contenderScore} $\gets 0$;\\
    \ForAll{\texttt{message} $\in$ \texttt{mailbox}}{
      \eIf{\texttt{message.target} = \texttt{ID}}{
        \If{\texttt{message.action} = FLIP}{
          \texttt{BECOME\_ROOT();}\label{line:becomeRoot}
        }
        \texttt{ADOPT\_CHILD(message);}\label{line:addChild}\hfill\comm{called for both FLIP or SELECT} 
      }{
        \If{\texttt{message.status\,=\,T $\wedge$ message.score $>$\,contenderScore}\label{line:mergingDetection}}{
          \texttt{contender $\gets$ message.ID;}\\
          \texttt{contenderScore $\gets$ message.score;}
        }
      }
    }

    \comm{Prepares the message to be sent}\smallskip
    \texttt{outMessage $\gets \bot$}\\
    \If{\texttt{status} $=$ T}{
      \eIf{\texttt{contenderScore $>$ score}}{
        \texttt{PREPARE\_MESSAGE(SELECT, contender)}; \label{line:prepareSelect}
      }{
        \If{\texttt{children} $\ne \emptyset$}{
          \texttt{PREPARE\_MESSAGE(FLIP, random(children))};\label{line:prepareFlip}
        }
      }
    }
    
    \If{\texttt{outMessage} $= \bot$}{
      \texttt{PREPARE\_MESSAGE(HELLO, $\bot$)}
    }
  }
  \caption{\label{algo:main}Main Algorithm}
\end{algorithm2e*}



\subsubsection{Regeneration:} The first thing a non-root node does after receiving the messages of the current round is to check whether the edge leading to its current parent is still present. If the edge disappeared, then the node regenerates a root directly (Line~\ref{line:become}). A nice property of the spanning forest principle is that this cannot happen twice in the same tree. And if a tree is broken into several pieces simultaneously, then each of the resulting subtree will have exactly one node performing this operation.

\subsubsection{The unique score technique:}
Unlike the high-level graph model from~\cite{CCGP13}, in which the merging operation involved two nodes in an {\em exclusive} way, the non-atomic nature of message passing allows for a {\em chain} of selection that may involve an arbitrary long sequence of nodes (e.g. $a$ selects $b$, $b$ selects $c$, and so on). This has both advantages and drawbacks. On the good side, it makes the initial merging process very fast (see rounds $1$ and $2$ in Figure~\ref{fig:example} to get an example). On the bad side, it is the reason why scores need to be introduced to avoid cycles. Indeed, relying only on a mere comparison of {\tt ID} to avoid cycles is not sufficient. Consider a chain of selection in round $i$ that ends up at some root node $u$. Nothing prevents $u$ to have passed the token to a lower-ID child, say $v$, in the previous round $i-1$ (that same round when $u$'s status $T$ was overheard by the next-to-last root in the chain). Now, nothing again prevents $v$ to have selected one of the nodes in the selection chain in round $i$, thereby creating a cycle. The score mechanism prevents such a situation by enforcing that after each FLIP, the new root has a larger score than its predecessor (see Lines~\ref{line:min-score} and~\ref{line:max-score} in Algorithm~\ref{algo:functions}). The score mechanism also guarantees that the current set of scores (network-wide) is always a permutation of the initial set of scores. Hence, scores are always unique. All of these elements are crucial ingredients in the proofs of correctness of Section~\ref{sec:correctness}.



\section{Correctness analysis}
\label{sec:correctness}

This section establishes a number of key properties about the spanning forest algorithm. In particular, we prove the claims regarding what property is {\em always} satisfied, regardless of the rate of changes. Because the proofs are technical, we provide first a preamble that includes helping definitions and a less technical outline of the proof. Then, the proof is described through two main parts called {\em consistency} and {\em correctness}, in reference to aspects defined in the preamble.

\subsection{Preamble and outline of the proof}
We first define a handful of instrumental concepts that help minimize the number of properties to be proven. Then, as we start formulating the key properties to be proved, we adopt concise notations regarding the state of the system. Precisely, we denote by $\uim.varname$ (resp. $\uip.varname$) the value of variable $varname$ at node $u$ before (resp. after) 
round $i$. Notice that for any node $u$, round $i$, and variable $varname$, we have $\uip.varname=\uipm.varname$. 
We use whichever notation is the most convenient in the given context.

\subsubsection{Helping definitions}

These definitions are not specific to our algorithm, they are general graph concepts that simplify the subsequent proofs.

\begin{definition}[Pseudotree and pseudoforest]
\label{def:pseudoforest}
A directed graph whose vertices have outdegree at most 1 is a {\em pseudoforest}.
  A vertex whose outdegree is 0 is called a {\em root}.
  The weakly connected components of a pseudoforest are called {\em pseudotrees}.
\end{definition}

\begin{lemma}
\label{lem:one_root}
  A pseudotree has at most one root.
\end{lemma}

\begin{proof}
   By definition, a pseudotree $\T=(V_\T, E_\T)$ is connected, thus $\vert E_\T \vert \geq \vert V_\T \vert - 1$. If $\T$ has several roots, then at least two nodes in $V_\T$ have no outgoing edge. Since the others have at most one, we must have $\vert E_\T \vert \leq \vert V_\T \vert - 2$, which is a contradiction.
 \end{proof}

\begin{lemma}
\label{lem:pseudotree_tree}
If a pseudotree $\T$ contains a root $r$, then it has no cycle.
\end{lemma}

\begin{proof} Let $V_1 \subset \T$ be the set of nodes at distance $1$ from $V_0=\{r\}$. Since $r$ has outdegree $0$, there is an edge from each node in $V_1$ to $r$. Since $\T$ is a pseudotree, these nodes have no other outgoing edge than those ending up in $V_{0}$. The same argument can be applied inductively, all nodes at distance $i$ having no other outgoing edges than those ending up in $V_{i-1}$.
 \end{proof}

\begin{definition}[Correct tree and correct forest]
\label{def:forest}
At the light of Lemma~\ref{lem:one_root} and~\ref{lem:pseudotree_tree}, we define a {\em correct tree} (or simply a {\em tree}) as a pseudotree in which a root can be found. We naturally define a {\em correct forest} (or simply a {\em forest}) as a pseudoforest whose pseudotrees are trees.

\end{definition}

Finally, because forests are considered in a spanning context, we say that a pseudoforest $\F$ is a correct forest {\em on graph} $G$ iff $\F$ is a correct forest {\em and} $\F$ is a subgraph of $G$. Defining correct trees as pseudotrees in which a root can be found is the key. When the moment arrives, this will allow us to reduce the correctness of our algorithm to the presence of a root in each pseudotree.


\subsubsection{Consistency}

At the end of a round, the state of an edge (whether it belongs to a tree, and if so, in what direction) must be consistently decided at both endpoints:

\begin{definition}[forest consistency]
The configuration $C_i$
is forest consistent if and only if for all nodes $u$,
$\uip.parent = v \Leftrightarrow u \in \vip.children$.
\end{definition}

The proof of forest consistency is inductively established by Theorem~\ref{th:consistency}, based on consistency of the initial configuration (Lemma \ref{lem:C0-forestConsistency}) and the maintenance the consistency over the rounds (Lemma \ref{lem:consistency}).
Forest consistency allows us to reduce the output of interest of the algorithm after each round $i$ to the mere \texttt{parent} variable. 

At the end of round $i$, the values of all \texttt{parent} variables should be consistent with the underlying graph $G_i$.

\begin{definition}[graph consistency]
The configuration $C_i$
is graph consistent if and only if for all nodes $u$,
$\uip.parent = v \Rightarrow \{u,v\} \in E_i$.
\end{definition}

This property is established by Corollary \ref{co:consistency}.
Graph consistency allows us to say that the output of the algorithm forms a pseudoforest on $G_i$.

\begin{definition}[Resulting forest]
  \label{def:F_i}
  Given a round $i\ge 1$, occurring on graph $G_i$, the graph $\mathcal{F}_i=(V,E_{\mathcal{F}_i})$ such that $E_{\mathcal{F}_i}=\{(u,v) : \{u,v\} \in E_i, \uip.parent=v\}$ is called the {\em pseudoforest} resulting from round~$i$. 
\end{definition}

\FloatBarrier

As explained in Section~\ref{sec:variables}, the variables {\tt parent} and {\tt status} are somewhat redundant, since the possession of a token is synonymous with being a root. The equivalence between both variables after each round is established in Lemma~\ref{lem:L5} (state consistency). The main advantage of this equivalence is that it allows us to formulate and prove a large number of lemmas using either variable, depending on which is the most convenient in the given context.

\subsubsection{Outline of the proof}

In this section, we prove that the resulting forest is always correct (Definition \ref{def:forest}). 
To achieve that goal, 
we first define a validity criterion at the node level, which recursively ensures the correctness of the pseudotree this node belongs to thanks to Definition~\ref{def:forest} ({\it i.e.} the existence of a root implies correctness).


\begin{definition}
A node $u$ is said to be valid at the beginning of round $i$ if either $\uim.status=T$ or $\uim.parent$ is valid. 
\end{definition}

\begin{algorithm2e}[h]
  
  
    
  
	\SetKwBlock{Begin}{}{}

        \texttt{procedure BECOME\_ROOT}
  \Begin{
  \texttt{status} $\gets$ T\;
  \texttt{parent} $\gets \bot$\;
  } 
  \bigskip

  \texttt{procedure ADOPT\_PARENT(outMessage)}\Begin{
  \texttt{status $\gets$ N}\;
  \texttt{parent $\gets$ outMessage.target}\;
  \If{\texttt{outMessage.action = FLIP}}{
    \texttt{children $\gets$ children$\smallsetminus$parent\;}
    \texttt{score $\gets$ min(score, mailbox[parent].score)\;\label{line:min-score}}
  }
  }
  \bigskip

  \texttt{procedure ADOPT\_CHILD(message)}
  \Begin{
  \texttt{children.add(message.ID)\;}
  \If{\texttt{message.action = FLIP}}{
    \texttt{score $\gets$ max(score, message.score)\;\label{line:max-score}}
  }
  }
  \bigskip

  \texttt{procedure PREPARE\_MESSAGE(action, target)}
  \Begin{
  \Switch{\texttt{action}}{
    \Case{\texttt{SELECT}}{
      \texttt{outMessage $\gets$ (ID, N, SELECT, target, score)}\; 
}
  \Case{\texttt{FLIP}}{
  \texttt{outMessage $\gets$ (ID, T, FLIP, target, score)}\;
  }
  \Case{\texttt{HELLO}}{
  \texttt{outMessage $\gets$ (ID, status, $\bot$, $\bot$, score)}\;
  }
  }
  }
  \caption{Functions called in Algorithm~\ref{algo:main}.}
  \label{algo:functions}
\end{algorithm2e}

The correctness of the whole forest can thus be established through showing that, first, it is initially correct (Lemma~\ref{lem:initial_forest}) and, second, if it is correct after round $i$, then it is correct after round $i+1$ (Theorem~\ref{lem:nodes_validity}). The latter is difficult to prove, and it involves a number of intermediate steps that correspond to a case analysis based on every action a node can perform (sending FLIP messages, SELECT messages, etc.).

We first prove that a node $u$ that sends a successful FLIP to $v$
in a round, is valid at the end of that round
(Lemma \ref{lem:FLIP-valid}) because at the end of that round $v$ is a root.
The proof relies on the fact that during a given round, a node cannot receive a FLIP and send a SELECT or a FLIP
(Lemma  \ref{lem:2FLIPs}). 

We then prove some necessary properties on the {\tt score} variable at each node.
For instance, a node changes its score at most once during a round
(Lemma \ref{lem:NO-FLIP-SCORE} and \ref{lem:change_score-once}).
Also, the set of all scores are a permutation of the node identifiers after each round (Lemma \ref{lem:permutation}).

Then we prove that a node that sends a successful SELECT in a round $i$, is valid at the end of that 
round (Lemma~\ref{lem:select_valid}). This part is the most technical and is the one that proves that chains of selection can not create cycles thanks to the property that score variables remain a permutation of all nodes IDs.

Finally, we prove that all roots at the beginning of a round
are still valid at the end of the round (Lemma \ref{lem:T-vald}).
Therefore, if all nodes are valid at the beginning of round, 
then they are also valid at the end of the round 
(Theorem \ref{lem:nodes_validity}). Since they are initially valid (Lemma~\ref{lem:initial_forest}), we conclude by induction on the number of rounds.

\subsection{Consistency (detailed proofs)}
\begin{lemma}
  \label{lem:C0-forestConsistency}
  \label{lem:initial_forest}
The configuration $C_0$ is forest consistent and graph consistent.
In $C_0$, the resulting pseudoforest is correct.
\end{lemma}
\begin{proof}
  The {\tt parent} variable is initialized to $\bot$. So, the configuration $C_0$ is forest consistent and graph consistent.
Any node $u$ belonging to the pseudotree $\T_u=(\{u\},\emptyset)$. 
Each of these pseudotrees contains  a root ($u$ itself) and is therefore a correct tree.
 \end{proof}



We say that $u$ {\em sends a FLIP} (resp. {\em SELECT}) in round $i$ if and only if $\uim.outMessage.action=FLIP$ (resp. {\em SELECT}). We say that it sends it {\em to node} $v$ if and only if $\uim.outMessage.target=v$. Finally the FLIP or SELECT is said to be {\em successful} (resp. {\em failed}) if $\{u,v\}\in E_i$ (resp. $\{u,v\}\notin E_i$).
 
\begin{lemma}[state consistency]
\label{lem:L5}
For all round $i \geq 0$, for all node $u$,
$\uip.status = T \Leftrightarrow \uip.parent = \bot$
\end{lemma}

\begin{proof}
Initially, at any node $u$, $u.status=T$ and $u.parent=\bot$. The change of $u.status$ to $N$ always comes with the assignment of a non-null $identifier$ ($outMessage.target$) to $u.parent$ (procedure \texttt{ADOPT\_PARENT()}), and assigning the value $T$ to $u.status$ is always followed by the change of $u.parent$ to $\bot$ (procedure \texttt{BECOME\_ROOT()}). So at any configuration, $v.parent=\bot$ if and only if $v.status=T$.
 \end{proof}

\begin{lemma}
  \label{lem:FLIP-PROCEDURE}
  If $u$ does not send a FLIP or SELECT in round $i$, then $u$ does not execute
the procedure \texttt{ADOPT\_PARENT()} during round $i$.
\end{lemma}

\begin{proof}
The execution of the procedure \texttt{ADOPT\_PARENT()} by $u$
is conditioned by the sending of a SELECT or a FLIP by $u$ during the current round (line \ref{line:testadoptparent}).
 \end{proof}

\begin{observation}
\label{obs:message_pepared}
At time where a node $u$ prepares its message to be
sent during the round $i$, we have $u.parent = \uimp.parent$
(resp. $children, status$).
\end{observation}

\begin{lemma}
  \label{lem:FLIP-SELECT-T}
  If $u$ sends a FLIP or SELECT in round $i$, then $\uim.status = T$.
\end{lemma}

\begin{proof}
$u$ sends in round $i$ the message prepared in round $i-1$. If $u$ sends a FLIP or a SELECT in round $i$ then in round $i-1$ \texttt{PREPARE\_MESSAGE()} is called  with FLIP or SELECT as action (lines \ref{line:prepareSelect} or \ref{line:prepareFlip}). Both instructions are conditioned by $status=T$.
 \end{proof}

\begin{lemma}
  \label{lem:T-if-T}
  If $v$ sends a message containing $T$ in round $i$, then $\vim.status=T$.
\end{lemma}
\begin{proof}
The procedure \texttt{PREPARE\_MESSAGE()} is executed by a node $u$ in round $i-1$ to construct 
the message $m$ to be sent in round $i$.
In all cases \texttt{PREPARE\_MESSAGE()} sets $m.senderStatus$ to $T$ only if $u.status=T$.
 \end{proof}

\begin{lemma}
\label{lem:select-score}
If $u$ sends a SELECT to $v$ in round $i$, then 
$\uim.score < \vimm.score$.
\end{lemma}

\begin{proof}
The value of the $score$ field in the message sent by a node $v$ in round $i-1$ 
is $\vimm.score$.

Assumes that the node  $u$ sends a SELECT to $v$ in a round $i$.
So, during the round $i-1$, $u$ sets its $contender$ variable to $v$ and its $contenderScore$ 
variable to $message.score$
$message$ being the message sent by $v$ at the begining of round $i-1$.
From that time to the end of round $i-1$, $u.score$ is not modified.

So $\uim.score < \vimm.score$, if $u$ sends 
a SELECT to $v$ in a round $i$.
 \end{proof}

\begin{lemma}
  \label{lem:cond-single-Flip}
If at the beginning of round $i$, the configuration 
is forest  consistent then only 
$\uim.parent$ can  send a FLIP at destination of $u$ during the round $i$.
\depends{\ref{th:consistency}}
\end{lemma}

\begin{proof}
A node $v$ can prepare a FLIP message to the node $u$ 
at then end of round $i-1$ only if $u \in \vim.children$.
We have $\uim.parent=v$ according to the hypothesis (forest consistency at the beginning of round). 
Therefore, only the node  $\uim.parent$ can prepare 
a FLIP message at destination of $u$, 
at the end of round $i-1$.
 \end{proof}

\subsubsection{Graph consistency:}
\begin{lemma}
  \label{lem:add_parent}
Let $u$ be a node such that $\uim.parent \neq v \wedge \uip.parent = v$.
Then $u$ sends a successful FLIP or SELECT to $v$ during the round $i$.
\end{lemma}
\begin{proof}
The only change of $parent$ by $u$ to a non-null identifier $v$ in 
a round $i$ is at the execution of the procedure \texttt{ADOPT\_PARENT()} 
which is conditioned by the reception of a message from $v$ (line \ref{line:becomeOrdinary}).
If $u$ receives the message of $v$ during round $i$ then $v$ effectively receives the
message sent by $v$ (\textit{reciprocal reception property}).
 \end{proof}

\begin{lemma}
  \label{lem:keep_parent}
Let $u$ be a node such that $\uim.parent = v \wedge \uip.parent = v$.
We have $\{u,v\} \in E_i$.
\end{lemma}
\begin{proof}
By Lemma \ref{lem:L5}, we have $\uim.status=N$.
So, $u$ does not send a FLIP or SELECT during the round $i$ 
(Lemma \ref{lem:FLIP-SELECT-T}).
Then, $u$ does not execute \texttt{ADOPT\_PARENT()} during the round $i$
according to Lemma  \ref{lem:FLIP-PROCEDURE}.
Since $\uip.parent = v$ we conclude that
$u$ does not execute the procedure \texttt{BECOME\_ROOT()} during the round $i$.
So $u$ did receive a message from $\uim.parent$ in round $i$.
We have $\{u,v\} \in E_i$.
 \end{proof}

\begin{corollary}[graph consistency]
\label{co:consistency}
Every configuration is graph consistent.
\end{corollary}
\begin{proof}
The configuration reached after any round is graph consistent
(Lemmas \ref{lem:add_parent} and \ref{lem:keep_parent}).
 \end{proof}

\subsubsection{Forest consistency:}

\begin{lemma}
\label{lem:N-parent}
If $\uim.parent=v$ then $\uip.parent=v$ or $\uip.parent=\bot$.
\depends{ }
\end{lemma}

\begin{proof}
According to Lemma \ref{lem:L5}, we have $\uim.status=N$, so 
u cannot send a  FLIP or a SELECT in round $i$ (by Lemma \ref{lem:FLIP-SELECT-T}).
Therefore, $u$ does not execute \texttt{ADOPT\_PARENT()} in round $i$ 
(Lemma \ref{lem:FLIP-PROCEDURE}). We conclude that 
$\uip.parent=v$ or $\uip.parent=\bot$.
 \end{proof}

\begin{lemma}
\label{lem:cond-2FLIPs}
Assume that at the beginning of round $i$, the configuration is forest consistent. If $u$ receives a FLIP in round $i$, then it does not send a FLIP nor a SELECT in round $i$.
\depends{\ref{th:consistency},  \ref{lem:FLIP-SELECT-T}, \ref{lem:L5}}
\end{lemma}

\begin{proof}
We will establish the contraposition of the lemma statement: 
if $u$ sends a FLIP or a SELECT in round $i$, then it does not receive a FLIP in round $i$.
By Lemma \ref{lem:FLIP-SELECT-T}, we have $\uim.status=T$.
According to Lemma \ref{lem:L5}, $\uim.parent=\bot$.
Thus according to the hypothesis (forest consistency at the beginning of round), for any node $v$, $u \notin \vim.children$.
Therefore no node has prepared a FLIP message at destination of $u$, 
in round $i-1$.
So $u$ cannot receive a FLIP in round $i$.
 \end{proof}

\begin{lemma}
 \label{lem:add_parent_add_child}
Assume that at the beginning of round $i$, the configuration is forest consistent. If in round $i$, $u$ changes $u.parent$ to $v$ then 
$u \in \vip.children$ : 
$\uim.parent \neq v \wedge \uip.parent = v \Rightarrow u \in \vip.children$.
\end{lemma}
\begin{proof}
$u$ sets $u.parent$ to $v$ only if the FLIP or SELECT was successful 
(Lemma \ref{lem:add_parent}). 
Therefore $v$ has received the FLIP or SELECT message sent by $u$.
 
The addition of a node $u$ to $v.children$ by $v$ is done during
the excution of the procedure \texttt{ADOPT\_CHILD()} which is 
conditioned by the reception of a FLIP or a SELECT message $m_u$ from $u$  
($m_u.target=v$, line \ref{line:addChild}). 
The procedure \texttt{ADOPT\_CHILD()} is executed after line \ref{line:intersection}
which is the only instruction that could remove $u$ from $v.children$. 
So, $u \in \vip.children$. 
We have $\uim.parent \neq v \wedge \uip.parent = v 
\Rightarrow u \in \vip.children$.
 \end{proof}

\begin{lemma}
  \label{lem:add_child_add_parent}
Assume that at the beginning of round $i$, the configuration is forest consistent. If in round $i$, $v$ adds $u$ to $v.children$ then $\uip.parent = v$ : $u \not \in \vim.children \wedge u \in \vip.children \Rightarrow \uip.parent = v$.
\depends{\ref{lem:cond-2FLIPs}}
\end{lemma}
\begin{proof}
 $v$ adds $u$ to $v.children$ only if it excutes the procedure 
\texttt{ADOPT\_CHILD()}
 which is conditioned by the reception of a FLIP or a SELECT 
sent by $u$.
As the reception of messages is reciprocal, 
$u$ also receives in round $i$ a message from $v$. This satisfies the condition for $u$ to execute the procedure 
\texttt{ADOPT\_PARENT()} which sets $u.parent$ to $v$.

Only the execution of \texttt{BECOME\_ROOT()} (at line \ref{line:becomeRoot}) 
could modify the value of $u.parent$.
This procedure would be executed only if $u$ has received
 a FLIP during round $i$ which cannot be the case.
 Notice that $u$ does not receive a FLIP during the round $i$ (Lemma \ref{lem:cond-2FLIPs}).
 \end{proof}

\begin{lemma}
  \label{lem:rem_parent_rem_child}
Assume that at the beginning of round $i$, the configuration is forest consistent. If in round $i$, $u$ changes $u.parent$ from $v$ to another value then 
$u \not \in \vip.children$ : 
$\uim.parent = v \wedge \uip.parent \neq v \Rightarrow u \not 
\in \vip.children$.
\end{lemma}

\begin{proof}
If $u$ changes $\uip.parent$ then we have $\uip.parent = \bot$ (Lemma \ref{lem:N-parent}).
Only the execution of \texttt{BECOME\_ROOT()} by $u$ sets $u.parent$ 
to $\perp$. The procedure \texttt{BECOME\_ROOT()} is executed
in two cases: at the detection of a disconnection (line \ref{line:become}), and
at the reception of a FLIP message (line \ref{line:becomeRoot}).

In the first case, the \textit{reciprocal reception property} ensures that $v$ 
does not receive the message sent by $u$.
So, $v$  removes $u$ from $children$ (line \ref{line:intersection}).

In the second case, $u$ receives a FLIP from $\uim.parent$ (Lemma 
 \ref{lem:cond-single-Flip}).
According to the \textit{reciprocal reception property}, $v$ receives the message sent by $u$ during the round $i$.
So, $v$ executes \texttt{ADOPT\_PARENT($(i^-)v.outMessage$)} which removes 
$u$ (i.e. $(i^-)v.outMessage.target$) from $v.children$ (line \ref{line:becomeOrdinary}).
 \end{proof}

\begin{lemma}
  \label{lem:rem_child_rem_parent}
Assume that at the beginning of round $i$, the configuration is forest consistent.  
 If in round $i$, $v$ removes $u$ from $v.children$ 
then $\uip.parent \neq v$ : $u \in \vim.children \wedge u 
\not \in \vip.children \Rightarrow \uip.parent \neq v$.
\end{lemma}
\begin{proof}
$v$ removes $u$ from $v.children$ in two cases: 
at the detection of a disconnection 
($v$ does not receive a message from $u$, line \ref{line:intersection}), and
when $v$ executes (\texttt{ADOPT\_PARENT($(i⁻).v.outMessage$)}, line \ref{line:becomeOrdinary}) 

In the first case, the \textit{reciprocal reception property} ensures that 
$u$ does not receive the message sent by $v$ during the round $i$.
So, $u$ becomes a root : it executes the procedure
\texttt{BECOME\_ROOT()} (line \ref{line:become}).

In the second case, $v$ executes \texttt{ADOPT\_PARENT($(i⁻).v.outMessage$)}.
So $v$ did send a successful FLIP or SELECT (Lemma \ref{lem:FLIP-PROCEDURE}).
As $v$ removes $u$ from $v.children$ during the execution of
\texttt{ADOPT\_PARENT($(i⁻).v.outMessage$)},
  we have $(i^-).v.outMessage.target=u$ and
$(i^-).v.outMessage.action=FLIP$ (see the procedure \texttt{ADOPT\_PARENT($outMessage$)}).
So $v$ sends
a successful FLIP to $u$  during round $i$.
Therefore, in round $i$, $u$ 
executes the procedure \texttt{BECOME\_ROOT()} (line \ref{line:becomeRoot}):
$u$ sets $u.parent$ to $\perp$.
 \end{proof}

\begin{lemma}[Forest Consistency]
\label{lem:consistency}
Let $i$ be a round starting from a forest consistent configuration.
The configuration reached at the end of round $i$ is forest consistent
\depends{\ref{lem:add_parent_add_child}, \ref{lem:add_child_add_parent}, \ref{lem:rem_parent_rem_child}, \ref{lem:rem_child_rem_parent}}
\end{lemma}
\begin{proof}
The configuration after the round $i$ is forest consistent according 
to Lemmas \ref{lem:add_parent_add_child}, \ref{lem:add_child_add_parent}, 
\ref{lem:rem_parent_rem_child}, \ref{lem:rem_child_rem_parent}. 
Notice that in the case where $u$ does not change the value of its parent variable (\textit{resp}. 
$u$ stays  in $v.children$) during round $i$, at the end of round $i$ the forest consistency 
property is preserved according to the contraposition of Lemma \ref{lem:rem_child_rem_parent} 
(\textit{resp}. contraposition of Lemma \ref{lem:rem_parent_rem_child}) and the hypothesis.
 \end{proof}

\begin{theorem}[Consistency]
\label{th:consistency}
Every configuration is forest consistent.
\end{theorem}
\begin{proof}
$C_0$ is forest consistent (Lemma \ref{lem:C0-forestConsistency}).
The configuration reached after any round is 
forest consistent (Lemma \ref{lem:consistency}).
 \end{proof}


%
%

\subsection{Correctness (detailed proofs)}

\subsubsection{Correctness of the resulting forest after token circulation:}

\begin{lemma}
  \label{lem:single-Flip}
  Let $v$ be a node. Only 
$\vim.parent$ can  send a FLIP at destination of $v$ during the round $i$.
\depends{\ref{th:consistency}, \ref{lem:cond-single-Flip})}
\end{lemma}

\begin{proof}
At the beginning of round $i$, the configuration is forest consistent
(Theorem \ref{th:consistency}).
Therefore, only the node  $\vim.parent$ can prepare 
a FLIP message at destination of $v$, 
at the end of round $i-1$ (Lemma \ref{lem:cond-single-Flip}).
 \end{proof}

\begin{lemma}
 \label{lem:2FLIPs}
 If $u$ receives a FLIP in round $i$, then it does not 
send a FLIP nor a SELECT in round $i$.
\depends{\ref{th:consistency},  \ref{lem:FLIP-SELECT-T}, \ref{lem:L5}}
\end{lemma}
\begin{proof}
At the beginning of round $i$, the configuration is forest consistent
(Theorem \ref{th:consistency}).
Therefore no node has prepared a FLIP message at destination of $u$, 
in round $i-1$ (Lemma \ref{lem:cond-2FLIPs}).
 \end{proof}

\begin{lemma}[Adoption]
\label{lem:adoption}
If $u$ sends a successful FLIP or SELECT to $v$ in round $i$, then $\uip.status = N$ and $\uip.parent=v$.
\depends{\ref{lem:2FLIPs}}
\end{lemma}

\begin{proof}
In round $i$, $u.outMessage.action=FLIP$ or SELECT  
and $v \in \uip.neighbors$. 
During the round $i$, $u$ executes the procedure \texttt{ADOPT\_PARENT()} 
(line \ref{line:becomeOrdinary})  
which sets $\uip.parent$ to $v$. 
According to Lemma \ref{lem:2FLIPs}, $u$ did not receive any FLIP 
message during the round $i$.
Only an execution of \texttt{BECOME\_ROOT()} by $u$ 
at line \ref{line:becomeRoot} can 
change the value of $u.parent$ 
during the round $i$. This line is not executed during round $i$.
 \end{proof}

\begin{lemma}
\label{lem:received-FLIP}
  If $u$ sends a successful FLIP to $v$, then $\vip.status = T$.
\depends{\ref{lem:2FLIPs}, \ref{lem:single-Flip}}
\end{lemma}

\begin{proof}
$v$ received $mu$ in round $i$, so $\{u,v\}\in E_i$.
$v$ executes the procedure \texttt{BECOME\_ROOT()} 
that changes $v.status$ to $T$. 
After the execution of line \ref{line:becomeOrdinary}, 
no instruction can set $v.status$ to $N$
until the end of round $i$. So $\vip.status = T$.
 \end{proof}

\begin{lemma}
  \label{lem:FLIP-valid}
  If $u$ sends a successful FLIP in round $i$, then $u$ is valid after round $i$.
\depends{\ref{lem:adoption}, \ref{lem:received-FLIP}}
\end{lemma}

\begin{proof}
By Lemmas \ref{lem:adoption} and \ref{lem:received-FLIP} $u$'s parent has a status $T$ after round $i$.
 \end{proof}

\subsubsection{Proofs on score permutations:}

\begin{lemma}
\label{lem:received-FLIP-score}
  If $u$ sends a successful FLIP to $v$, then $\uim.score \leq \vip.score$.
\depends{\ref{lem:2FLIPs}, \ref{lem:single-Flip}}
\end{lemma}

\begin{proof}
$u$ sent a message $mu$ to $v$ at the beginning of round $i$ 
such that $mu.action=$ FLIP,  $mu.target=v.ID$ and $mu.score = \uim.score$. 
$v$ received $mu$ in round $i$, so $\{u,v\}\in E_i$.
$v$ executes the procedure
\texttt{ADOPT\_CHILD($mu$)} at line \ref{line:addChild} 
in round $i$. 
This procedure sets the current score of $v$ 
to $max(v.score,mu.score)$, as $mu.score=\uim.score$.
After the execution of this instruction, we have
$mu.score=\uim.score \leq v.score$. 
We notice that after this operation, no instruction 
can change the value of $v.score$  
(Lemma \ref{lem:single-Flip}.
 \end{proof}

\begin{lemma}
\label{lem:NO-FLIP-SCORE}
$\uim.score = \uip.score$ unless $u$ sends 
or receives a successful FLIP in round $i$.
\depends{\ref{lem:change_score}}
\end{lemma}
\begin{proof}

$u$ changes its $score$ value only by executing \texttt{ADOPT\_PARENT($m_u$)} or \texttt{ADOPT\_CHILD($m_u$)}. Both instructions that changes $u.score$ value in these procedures (Algorithm \ref{algo:functions}, line \ref{line:becomeOrdinary}, \ref{line:addChild}) are conditioned by $m_u.action=FLIP$.
%
%
%
 \end{proof}

\begin{lemma}
  \label{lem:change_score-once}
A node $u$ changes $u.score$ at most once during a round.
\depends{\ref{lem:NO-FLIP-SCORE}, \red{lem:single-Flip}, \ref{lem:2FLIPs}}
\end{lemma}

\begin{proof}
A node sends at most one  FLIP message during a round.
A node receives at most one FLIP message during a round 
(Lemma \ref{lem:single-Flip}).
Either a node receives a FLIP, sends one, or it does not receive and does not send a FLIP during a given
round (Lemma \ref{lem:2FLIPs}).
So, according to Lemma \ref{lem:NO-FLIP-SCORE}, 
a node changes $u.score$ at most once during a round.
 \end{proof}

\begin{lemma}
  \label{lem:permutation}
  Before each round, the set of scores is a permutation of the set of identifiers. 
\depends{\ref{lem:change_score_once}}
\end{lemma}

\begin{proof}
After the initialization in each node $u$, $u.score=u.ID$. A node $u$ changes its score only by executing \texttt{ADOPT\_PARENT()} or \texttt{ADOPT\_CHILD()}.
We will do a proof by induction.
We assume at the beginning of round $i$, the set of scores is a a permutation of the set of indentifiers.
We have for any node $u$, $mu.score=\uim.score$.

According to Lemma \ref{lem:NO-FLIP-SCORE}, only a node sending or receiving a successful FLIP may change its $score$ value. 
Assume that the node $u$ changes its $score$ value during round $i$. Without lost of generality, we assume $u$
sends the successful FLIP to a node $v$ in round $i$.

By hypothesis, $u$ changes its $score$ to $\vim.score$ during the execution of \texttt{ADOPT\_PARENT()} in round $i$. We have $\uim.score \ geq \vim.score$. 
$v$ executes the procedure
\texttt{ADOPT\_CHILD($mu$)} at line \ref{line:addChild} 
in round $i$. 
This procedure sets the current score of $v$ 
to $max(v.score,mu.score)$, as $mu.score=\uim.score$.
After the execution of this instruction, we have
$ v.score=\uim.score$. 

According lemma \ref{lem:change_score-once}, we have 
$\vip.score=\uim.score$ and $\uip.score=\vim.score$. 
 \end{proof}

\subsubsection{Correctness of the resulting forest after mergings:}

In lemmas \ref{lem:selected-T} and  \ref{lem:selected-N}, we establish that
if $u$ sends a successful SELECT to $v$ in round $i$ either $\vim.status= T$ or $(i^-)v.parent.status= T$. In the first case, we have $\uim.score < \vim.score$, and in the second case, 
we have $\uim.score < (i^-)v.parent.score$.
Let $ch$ be a series of nodes $u_0, u_1, u_2$ such that $(i^+)u_j.parent=u_{j+1}$  
and such that $u_0$ sends a successful SELECT to $u_1$ during the round $i$.
As a $ch$'s subchain of nodes having strictly increasing scores at the beginning of round $i$ 
may be built: $ch$ has not loop. So $ch$ ends by a node having a token : all nodes on that chain are valid.

%

%

\begin{lemma}
  \label{lem:send-T-become-N-parent-T}
  If $v$ sends a message containing $T$ in round $i$ and $\vip.status=N$, let $w=\vip.parent$, then $\wip.status=T$.
\end{lemma}

\begin{proof}
If $v$ sends a message containing $T$ in round $i$, then $\vip.status=T$. If $\vip.status=N$, then $v$ has executed \texttt{ADOPT\_PARENT()} in round $i$, because it is the only procedure that sets $v.status$ to $N$.
$v$ executes \texttt{ADOPT\_PARENT()} only if it has sent a FLIP message $m_v$ to a node $w$ ($m_v.action \neq$ SELECT because $m_v.senderStatus=T$), and if $w$ has received the message $m_v$ (\textit{reciprocal reception property}). At the reception of $m_v$ by $w$, $w$ executes \texttt{BECOME\_ROOT()} (line 16) which sets $w.status$ to $T$ and from this line until the end of the round no instruction can change $w.status$ to $N$. So $\wip.status=T$.

At the execution of \texttt{ADOPT\_PARENT()} by $v$, $v$ sets $v.parent$ to $w$. After this instruction there is only \texttt{BECOMES\_ROOT()} that can modifie the value of $v.parent$, and which is conditioned by the reception of a FLIP message. According to lemma \ref{lem:2FLIPs} $v$ cannot call \texttt{BECOMES\_ROOT()} because it cannot receive a FLIP message. So $w=\vip.parent$.

So, if $v$ sends a message containing $T$ in round $i$ and $\vip.status=N$, and $w=\vip.parent$, then $\wip.status=T$.
 \end{proof}

\begin{lemma}
\label{lem:sendT}
If $v$ sends a message containing $T$
in round $i$ and $\vip.status=N$, let $w=\vip.parent$, then $\wip.score \ge \vim.score$.
\depends{\ref{lem:received-FLIP}, \ref{lem:T-if-T}}
\end{lemma}

\begin{proof}
  We have $\vim.status=T$ because in round $i-1$, $v.status$ cannot be modified after the execution of \texttt{PREPARE\_MESSAGE()}. 
If $\vim.children \neq \emptyset$ 
then $v$ sends a FLIP message to one of its children, named $u$, in round $i$.
Either $\{u,v\}\in E_i$, then $\uip.parent=v$, $\vip.status = T$ and
$\uip.score \leq \vim.score$ 
(see Lemmas \ref{lem:received-FLIP} and \ref{lem:received-FLIP-score}).
Otherwise $\vip.status=T$.
 \end{proof}

\begin{lemma}
 \label{lem:selected-was-T}
If $u$ sends a successful SELECT to $v$ in round $i$ then $\vimm.status= T$.
\depends{\ref{lem:T-if-T}}
\end{lemma}
\begin{proof}
    Node $u$ prepared a SELECT message to $v$ in round $i-1$,
    thus it had $u.contender=v$, which implies it received from $v$ 
    a message containing $T$. We have then $\vimm.status= T$ because after the execution of \texttt{PREPERE\_MESSAGE()} by $v$ in round $i-2$, $v.status$ cannot be changed. 
 \end{proof}

\begin{lemma}
 \label{lem:selected-T}
If $u$ sends a successful SELECT to $v$ in round $i$ and $\vim.status= T$, then $\uim.score < \vim.score$.
\depends{\ref{lem:select-score}, \ref{lem:T-T-keep-score}, \ref{lem:selected-was-T}}
\end{lemma}

\begin{proof}
    By Lemma~\ref{lem:selected-was-T} $\vimm.status=T$. Then Lemmas~\ref{lem:select-score} and~\ref{lem:NO-FLIP-SCORE} respectively imply that $\uim.score < \vimm.score$ and $\vimm.score = \vim.score$.
 \end{proof}

\begin{lemma}
 \label{lem:selected-N}
If $u$ sends a successful SELECT to $v$ in round $i$ and  
$\vim.status= N$, then let $w=\vim.parent$. 
It holds that $\wim.status= T$ 
and $\uim.score < \wim.score$.
\depends{\ref{lem:select-score}, \ref{lem:sendT}, \ref{lem:selected-was-T}, \ref{lem:send-T-become-N-parent-T}}
\end{lemma}
\begin{proof}
    By Lemma~\ref{lem:selected-was-T} we have $\vimm.status=T$. 
    Then Lemmas~\ref{lem:select-score} and~\ref{lem:sendT} respectively imply that $\uim.score < \vimm.score$ and $\vimm.score \leq \wim.score$.  
Lemma \ref {lem:send-T-become-N-parent-T}  implies that $\wim.status= T$.
 \end{proof}

\begin{lemma}[Cancellation]
  \label{lem:CANCEL}
  If $u$ sends a failed FLIP or SELECT in round $i$, then $\uip.status = T$.
\depends{\ref{lem:FLIP-SELECT-T}}
\end{lemma}

\begin{proof}
By lemma \ref{lem:FLIP-SELECT-T}, we have $\uim.status = T$. 
$v$ did not receive the message from $u$ implies that $\{u,v\}\notin E_i$. 
So, in round $i$, $v \notin u.neighbors$ ($u$ did not receive 
the message from $v$).
Only during the execution of \texttt{ADOPT\_PARENT()},
called in line \ref{line:becomeOrdinary}, $u$ can change its $status$ to $N$.
This procedure is not executed during the round $i$.
 \end{proof}

\begin{lemma}[Conservation]
  \label{lem:conservation}
  If $\uim.status=T$ and $u$ does not send a FLIP or SELECT in round $i$, then $\uip.status=T$.
\end{lemma}

\begin{proof}
By lemma \ref{lem:FLIP-PROCEDURE}, $u$ does not execute the procedure \texttt{ADOPT\_PARENT()} during the round $i$.
$u$ can set \texttt{status} variable to $N$ only if it executes \texttt{ADOPT\_PARENT()}. 
 \end{proof}

\begin{lemma}
  \label{lem:all-but-SELECT-valid}
  If $\uim.status=T$ and $u$ does not send a successful SELECT in round $i$, then $u$ is valid after the round $i$.
\end{lemma}

\begin{proof}
According to Lemma \ref{lem:FLIP-valid}, after the successful sending of a FLIP message in round $i$, $u$ is valid at the end of round $i$. 
If $u$ sends a failed SELECT or a failed FLIP then $u$ is valid after the round $i$ by Lemma~\ref{lem:CANCEL}.
otherwise, $u$ did not send a SELECT or a FLIP during the round :
it is also valid at the end of the round by Lemma~\ref{lem:conservation}.
 \end{proof}

\begin{lemma}
\label{lem:select_valid}
If a node sends a successful SELECT in round $i$,
then it is valid at the end of round $i$. 
\depends{\ref{lem:permutation}, \ref{lem:regeneration}, \ref{lem:adoption}, \ref{lem:send-T-become-N-parent-T}, \ref{lem:selected-T}, \ref{lem:selected-N}, \ref{lem:all-but-SELECT-valid}}
\end{lemma}

\begin{proof}
Let $S$ be the set of nodes that send a successful  
SELECT in round $i$ and are not valid at the end of round $i$. We will prove, by contradiction, that $S$ is empty. Assume $S$ is non-empty and consider the node in $S$ that had the largest score at the beginning of round (say, node $u$). Such a node exists by Lemma~\ref{lem:permutation}. We will prove that $u$ is valid after the round, which is a contradiction. Let $v$ be the recipient of $u$'s successful SELECT. By Lemma~\ref{lem:adoption} $\uip.parent=v$, thus is enough to show that $v$ is valid after round $i$ to get our contradiction. Let us examine both cases whether $\vim.state= T$ or $N$.

If $\vim.status= T$, then either $v$ also sends a successful SELECT in round $i$, or it does not. If it does not, then it is valid after round $i$ (Lemma~\ref{lem:all-but-SELECT-valid}). 
If it does, then it must be valid otherwise $u$ is not maximal in $S$ (Lemma~\ref{lem:selected-T}).

If $\vim.status= N$, then let $w=\vim.parent$. Two cases are considered, whether $\{v,w\} \in E_i$ or not. If $\{v,w\} \notin E_i$ then $\vip.status=T$ because the condition forces $u$ to call the procedure \texttt{BECOME\_ROOT()} in line~\ref{line:become} which makes it take the status $T$. After, $u$ can takes the status $N$, only during the execution of the procedure \texttt{ADOPT\_PARENT()}
in line~\ref{line:becomeOrdinary}. 
This procedure is called by $u$ only 
if $u$ did send a FLIP or a SELECT 
at the beginning of round $i$ by lemma  \ref{lem:FLIP-PROCEDURE}. 
By Lemma~\ref{lem:FLIP-SELECT-T}, 
this cannot happen.
Thus $v$ is valid after round~$i$. If $\{v,w\} \in E_i$, we use the fact that $\wim.status= T$ (Lemma~\ref{lem:send-T-become-N-parent-T}) to apply the same idea as we did above: either $w$ also sends a successful SELECT in round~$i$, or it does not. If it does not, then it is valid after round~$i$ (Lemma~\ref{lem:all-but-SELECT-valid}). If it does, then it must be valid otherwise $u$ is not maximal in $S$ (Lemma~\ref{lem:selected-N}).
 \end{proof}

\subsubsection{Correctness of resulting forest:}

\begin{lemma}
\label{lem:T-vald}
If $\uim.status=T$ then $u$ is valid after round $i$. 
\depends{  \ref{lem:CANCEL}, \ref{lem:FLIP-valid}, \ref{lem:conservation}, \ref{lem:select_valid}}
\end{lemma}
\begin{proof}
According to Lemma \ref{lem:select_valid}, after the successful sending of a SELECT message in round $i$, $u$ is valid at the end of round $i$. 
According to Lemma \ref{lem:FLIP-valid}, after the successful sending of a FLIP message in round $i$, $u$ is valid at the end of round $i$. 
If $u$ sends a failed SELECT or a failed FLIP then $u$ is valid after the round by Lemma~\ref{lem:CANCEL}.
In otherwise, $u$ is also valid the round by Lemma~\ref{lem:conservation}.

 \end{proof}

\begin{theorem}[Resulting forest correctness]
  \label{lem:nodes_validity}
  If all nodes are valid at the beginning of the the round $i$, then 
all nodes are valid after round $i$.
\depends{ \ref{lem:T-vald}, \ref{lem:N-parent} }
\end{theorem}
\begin{proof}

Assume that  a node $v$ is invalid after round $i$.
According to Lemma \ref{lem:T-vald}, $\vim.status = N$.

Let $u_0, u_1, u_2, ..., u_k$ be the finite series of nodes such 
that for $j \in [0,k-1]$,
$\uim_j.parent=u_{j+1}$, $\uim_k.status=T$, and $u_0=v$.
This series exists because $u$ is valid at the beginning of round $i$.

Let $u'_1, u'_2, ..., $ be the infinite series of nodes such 
that for all $j \geq 1$
$\uip'_j.parent=u_{j+1}$,  and 
$\vip.parent=u'_{1}$. 
This series exists because $v$ is invalid (by hypothesis).

According to Lemma \ref{lem:N-parent}, $j \in [1,k]$, $u_j=u'_j$.
According to Lemma \ref{lem:T-vald}, $u_k$ is valid.
So all nodes of the series $u_0, u_1, u_2, ..., u_k$ are valid.
There is a contradiction.
 \end{proof}

\section{Convergence and preliminary discussion on performance}
\label{sec:simulation}

Our main focus in this paper was to present the spanning forest algorithm and prove that it guarantees a number of key properties, whatever the dynamics. Somewhat ironically, the same properties are satisfied even if the algorithm does nothing beyond initialization: every node remains forever a single-node tree, which satisfies all the predicates. Naturally, one expects more than this from an algorithm, which brings us to topics related to convergence and performance. We offer here a preliminary discussion on these topics, starting with what quality metric is adapted in such as context and how our algorithm behaves in this respect.

\subsection{What metric does make sense?}
\label{sec:metric}

The natural way to define an optimal (at least, irreductible) state in a spanning forest problem for partitioned networks is to have every connected component spanned by a {\em single} tree. However, even though 
one expects changes to obey some natural constraints depending on the mobility scenario, it is unreasonable to expect that an algorithm (however good it be) has sufficient time to converge towards an irreductible state in-between changes. Another remark is that the execution of algorithms for this type of problems never terminates; they are ever going.

In this context, a reasonable metric for evaluating our algorithm (or comparing two algorithms) is rather the average {\em number} of trees in each connected component, taken {\it e.g.} over the execution or in a stationary regime (if a stochastic model of dynamic networks is used for generating an infinite lifetime network). This being said, if the network {\em were} to stabilize, then one would indeed expect that a single tree to span each component. Both aspects are now discussed.

\subsection{Convergence in case of network stability} At an abstract level, the spanning forest algorithm presented in this paper relies on random walks in trees. Since trees are bipartite graphs, it may so happen that two tokens never meet (at both extremities of a common edge), although their trees could have been merged. 
Standard techniques exist for preventing periodic walks, such as stopping the tokens occasionally (also called lazy walks). This variant is easy to incorporate in the existing algorithm, by having a node decide whether or not circulating the token (FLIP messages) with some probability. Apart from this, markov chain theory tells us that (again, if the graph does not change) the tokens {\em will} eventually meet and thus every component will eventually be spanned by a single tree. The speed of this convergence relates to the area of coalescing random walks (see {\it e.g.}~\cite{Cooper13}), which is out of the scope of this paper (and beyond our technical skills).

\subsection{A practical scenario}
We verified the applicability of our algorithm in a real world
scenario. The algorithm was implemented\footnote{The source code of our algorithm is available upon request.} using the JBotSim library~\cite{C13} and 
tested against the {\tt Infocomm06} dataset~\cite{infocom06}. This well known dataset is a record of the communication links among devices given to
 $78$ people during the {\sc Infocomm} conference in 2006.
The update rate for the links is every 120 seconds, which means that
the presence time of an edge is a multiple of 120 seconds, a somewhat
optimistic value. 
To counterbalance this, we chose pessimistic options as to the number of 
rounds the nodes can perform in one second: 10 rounds (mildly pessimistic) or 1 round (very pessimistic).
In each case, we measured the average number of trees per connected
component over the execution (as discussed in Section~\ref{sec:metric}). The results are shown on Figures~\ref{fig:infocomm10} and~\ref{fig:infocomm1}, in which every point corresponds an average over 100 executions. 
\begin{figure}[h]
\center
\includegraphics[scale=0.45]{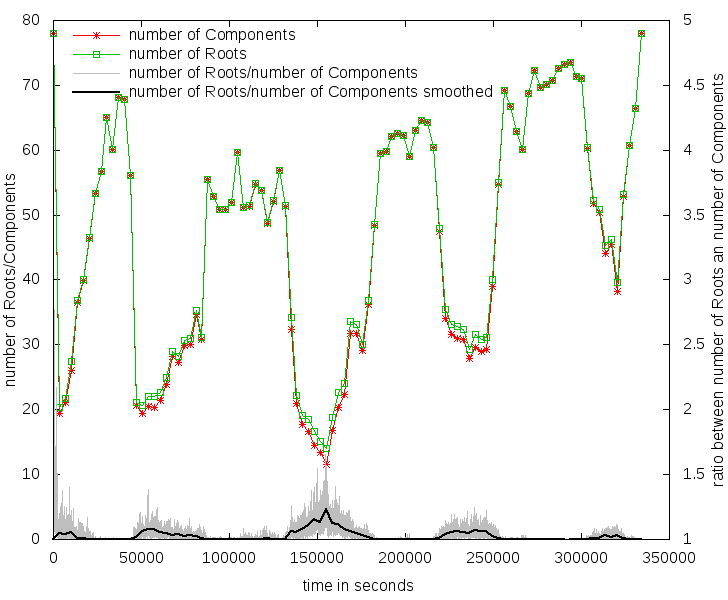}
\caption{Number of roots per connected components, assuming 10 rounds per second.}
\label{fig:infocomm10}
\end{figure}
\begin{figure}[h]
\center
\includegraphics[scale=0.45]{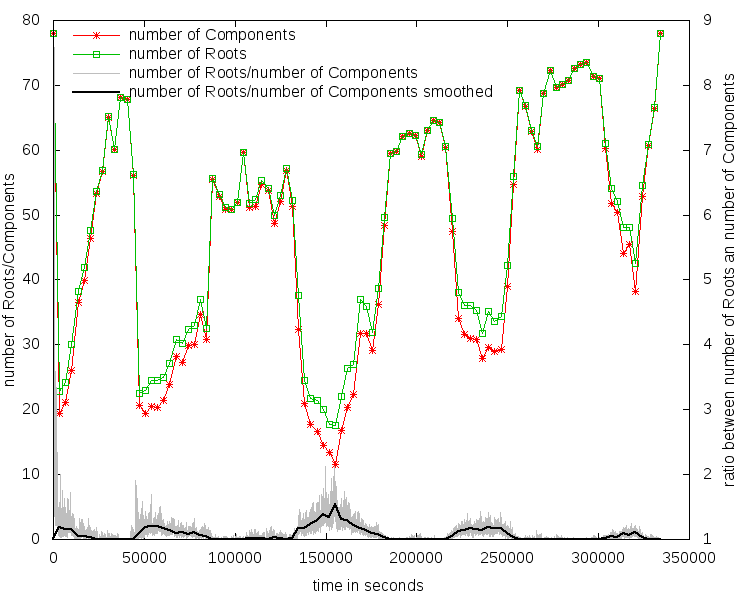}
\caption{Number of roots per connected components, assuming only 1 round per second.}
\label{fig:infocomm1}
\end{figure}

As one can see, the number of trees per connected component is often close to $1$ ($1.027$ in average in the first case, and $1.080$ in the second case). Furthermore, the algorithm achieves an optimal configuration of a single spanning tree per connected component  about $47\%$ of the time in the first case ($32.68\%$ in the second case), which implies that the algorithm may be relevant in practical scenarios and thus worth further investigation regarding its performance.

\section*{Acknowledgment}
{\small This work was partially supported by ANR projects ESTATE (ANR-16-CE25-0009-03) and DESCARTES (ANR-16-CE40-0023).
}

\bibliographystyle{plain}
\bibliography{spanning-forest}

\begin{thebibliography}{10}

\bibitem{mosbah-tree}
Sheila Abbas, Mohamed Mosbah, and Akka Zemmari.
\newblock Distributed computation of a spanning tree in a dynamic graph by
  mobile agents.
\newblock In {\em Proc. of IEEE Int. Conference on Engineering of Intelligent
  Systems (ICEIS)}, pages 1--6, 2006.

\bibitem{AF02}
David Aldous and Jim Fill.
\newblock Reversible markov chains and random walks on graphs, 2002.

\bibitem{AAD+06}
Dana Angluin, James Aspnes, Zo\"e Diamadi, Michael~J. Fischer, and Ren{\'e}
  Peralta.
\newblock Computation in networks of passively mobile finite-state sensors.
\newblock {\em Distributed Computing}, 18(4):235--253, 2006.

\bibitem{Awerbuch08}
Baruch Awerbuch, Israel Cidon, and Shay Kutten.
\newblock Optimal maintenance of a spanning tree.
\newblock {\em J. ACM}, 55(4):18:1--18:45, September 2008.

\bibitem{AE84}
Baruch Awerbuch and Shimon Even.
\newblock Efficient and reliable broadcast is achievable in an eventually
  connected network.
\newblock In {\em Proceedings of the third annual ACM symposium on Principles
  of distributed computing}, pages 278--281. ACM, 1984.

\bibitem{Baala03}
Hichem Baala, Olivier Flauzac, Jaafar Gaber, Marc Bui, and Tarek El-Ghazawi.
\newblock A self-stabilizing distributed algorithm for spanning tree
  construction in wireless ad hoc networks.
\newblock {\em Journal of Parallel and Distributed Computing}, 63:97--104,
  2003.

\bibitem{BZ89}
Judit Bar-Ilan and Dror Zernik.
\newblock Random leaders and random spanning trees.
\newblock In {\em Workshop on Distributed Algorithms (WDAG)}, volume 392 of
  {\em Lecture Notes in Computer Science}, pages 1--12. Springer Berlin
  Heidelberg, 1989.

\bibitem{BBS13}
Thibault Bernard, Alain Bui, and Devan Sohier.
\newblock Universal adaptive self-stabilizing traversal scheme: Random walk and
  reloading wave.
\newblock {\em J. Parallel Distrib. Comput.}, 73(2):137--149, 2013.

\bibitem{asynchronous}
Janna Burman and Shay Kutten.
\newblock Time optimal asynchronous self-stabilizing spanning tree.
\newblock In Andrzej Pelc, editor, {\em Distributed Computing}, volume 4731 of
  {\em Lecture Notes in Computer Science}, pages 92--107. Springer Berlin
  Heidelberg, 2007.

\bibitem{C13}
Arnaud Casteigts.
\newblock The {JBotSim} library.
\newblock {\em CoRR}, abs/1001.1435, 2013.
\newblock See also the project website at \url{http://jbotsim.sourceforge.net}.

\bibitem{CCGP13}
Arnaud Casteigts, Serge Chaumette, Fr{\'e}d{\'e}ric Guinand, and Yoann
  Pign{\'e}.
\newblock Distributed maintenance of anytime available spanning trees in
  dynamic networks.
\newblock In {\em Proceedings of 12th conf. on Adhoc, Mobile, and Wireless
  Networks (ADHOC-NOW)}, volume 7960 of {\em Lecture Notes in Computer
  Science}, 2013.

\bibitem{CFMS12}
Arnaud Casteigts, Paola Flocchini, Bernard Mans, and Nicola Santoro.
\newblock Shortest, fastest, and foremost broadcast in dynamic networks.
\newblock {\em CoRR}, abs/1210.3277, 2014.

\bibitem{Cooper13}
Colin Cooper, Robert Elsasser, Hirotaka Ono, and Tomasz Radzik.
\newblock Coalescing random walks and voting on connected graphs.
\newblock {\em SIAM Journal on Discrete Mathematics}, 27(4):1748--1758, 2013.

\bibitem{IJ90}
Amos Israeli and Marc Jalfon.
\newblock Token management schemes and random walks yield self-stabilizing
  mutual exclusion.
\newblock In {\em Proceedings of the ninth annual ACM symposium on Principles
  of distributed computing}, pages 119--131. ACM, 1990.

\bibitem{synchronous}
Alex Kravchik and Shay Kutten.
\newblock Time optimal synchronous self stabilizing spanning tree.
\newblock In Yehuda Afek, editor, {\em Distributed Computing}, volume 8205 of
  {\em Lecture Notes in Computer Science}, pages 91--105. Springer Berlin
  Heidelberg, 2013.

\bibitem{KLO10}
Fabian Kuhn, Nancy Lynch, and Rotem Oshman.
\newblock Distributed computation in dynamic networks.
\newblock In {\em Proceedings of the 42nd ACM symposium on Theory of computing
  (STOC)}, pages 513--522. ACM, 2010.

\bibitem{GRS01}
Igor Litovsky, Yves Metivier, and Eric Sopena.
\newblock Graph relabelling systems and distributed algorithms.
\newblock In {\em Handbook of graph grammars and computing by graph
  transformation}. Citeseer, 2001.

\bibitem{infocom06}
James Scott, Richard Gass, Jon Crowcroft, Pan Hui, Christophe Diot, and
  Augustin Chaintreau.
\newblock Crawdad trace cambridge/haggle/imote/infocom (v. 2006-01-31).
\newblock 2006.

\end{thebibliography}

\end{document}